\documentclass[aps,prl,twocolumn,superscriptaddress,groupedaddress]{revtex4}
\usepackage{graphicx}  % needed for figures
\usepackage{dcolumn}   % needed for some tables
\usepackage{bm}        % for math
\usepackage{dsfont}
\usepackage{amssymb}   % for math
\usepackage[utf8]{inputenc}
\usepackage[english]{babel}
\usepackage{amsthm}
\usepackage{amsmath}
\usepackage{comment}
\usepackage{cancel}
\usepackage[paperwidth=210mm,paperheight=297mm,centering,hmargin=2cm,vmargin=2.5cm]{geometry}
\usepackage[normalem]{ulem}
\usepackage{color}
\usepackage{braket}
\newtheoremstyle{mytheoremstyle} % name
    {\topsep}                    % Space above
    {\topsep}                    % Space below
    {}                   % Body font
    {}                           % Indent amount
    {\itshape}                   % Theorem head font
    {.}                          % Punctuation after theorem head
    {.5em}                       % Space after theorem head
    {}  % Theorem head spec (can be left empty, meaning ‘normal’)

\theoremstyle{mytheoremstyle}
\newtheorem{lemma}{Lemma}
\newtheorem{result}{Result}
\newtheorem{definition}{Definition}
\newtheorem{corollary}{Corollary}

\usepackage[svgnames]{xcolor}
\usepackage{tikz}
\usetikzlibrary{decorations.markings}
\usetikzlibrary{shapes.geometric}
\usetikzlibrary{circuits.ee.IEC}

\pgfdeclarelayer{edgelayer}
\pgfdeclarelayer{nodelayer}
\pgfsetlayers{edgelayer,nodelayer,main}

\tikzstyle{none}=[]
\tikzstyle{trace}=[circuit ee IEC,thick,ground,rotate=0,scale=2]
\begin{document}
\widetext

\title{Entanglement is necessary for emergent classicality in all physical theories}

\author{Jonathan G. Richens}\thanks{These authors contributed equally to this work}
\affiliation{Controlled Quantum Dynamics theory group, Department of Physics, Imperial College London, London SW7 2AZ, UK.}
\affiliation{Department of Physics and Astronomy, University College London,
Gower Street, London WC1E 6BT, UK.}

\author{John H. Selby\footnotemark[1]}
\affiliation{Controlled Quantum Dynamics theory group, Department of Physics, Imperial College London, London SW7 2AZ, UK.}
\affiliation{University of Oxford, Department of Computer Science, Oxford OX1 3QD, UK.}

\author{Sabri W. Al-Safi}
\affiliation{School of Science \& Technology, Nottingham Trent University, Burton Street, Nottingham NG1 4BU, UK}

\begin{abstract}
One of the most striking features of quantum theory is the existence of entangled states, responsible for Einstein's so called ``spooky action at a distance''. These states emerge from the mathematical formalism of quantum theory, but to date we do not have a clear idea of the physical principles that give rise to entanglement.
Why does nature have entangled states? Would any theory superseding classical theory have entangled states, or is quantum theory special? One important feature of quantum theory is that it has a classical limit, recovering classical theory through the process of decoherence. We show that any theory with a classical limit must contain entangled states, thus establishing entanglement as an inevitable feature of any theory superseding classical theory.
\end{abstract}

\maketitle

\section{Introduction}

Entanglement and non-locality are two of the features of quantum theory that clash most strongly with our classical preconceptions as to how the universe works. In particular, they create a tension with the other major theory of the twentieth century: relativity~\cite{einstein1916foundation}. This is most clearly illustrated by Bell's theorem~\cite{bell1966physics,wood2015lesson}, in which certain entangled states are shown to violate local realism by allowing for correlations that cannot be explained by classical causal structures \cite{wood2015lesson}.
In this paper we ask whether entanglement is a surprising feature of nature, or whether it should be expected in any non-classical theory? Could a scientist with no knowledge of quantum theory have predicted the existence of entangled states based solely on the premise that their classical understanding of the world was incomplete?

Any such scientist could reasonably postulate the existence of a classical regime - in that whatever theory describes reality must be able to behave like classical theory in some limit. Although this is a very natural assumption, given that we frequently observe systems behaving classically, we show that it imparts very strong constraints on the structure of any non-classical theory. Indeed in  \cite{landau1977quantum}, Landau and Lifshitz noted

\begin{quote}
``Quantum mechanics occupies a very unusual place among physical theories: it contains classical mechanics as a limiting case, yet at the same time it requires this limiting case for its own formulation''
\end{quote}

Thus to answer these questions, we explore all theories that have a classical limit \cite{zurek2009quantum,joos2013decoherence}. This is formalised in quantum theory by \emph{decoherence maps}, which take quantum systems to semi-classical states with respect to some basis. Physically, decoherence maps represents a quantum system interacting with some inaccessible environment resulting in the loss of quantum coherences. Inspired by this we develop a generalization of decoherence maps for arbitrary operationally defined theories (see \cite{gogioso2017categorical,selby2017leaks,lee2017no,coecke2017twoTEST} for a related categorical/process-theoretic approaches.). We consider all theories that can decohere to classical theory and show that any such theory either contains entangled states, or is classical theory. Thus, the existence of these classically counter-intuitive entangled states present in quantum theory can be understood as arising from, and being necessary for, the existence of a classical world. This result hints towards the possibility that other counter-intuitive features of quantum theory could be derived from its accommodation of a classical limit, and paves the way for deriving the features of post-classical and post-quantum theories from the existence of this limit.

The outline of this paper is as follows. In the following section we discuss the framework, describe the minimal characteristics expected of the generalised decoherence map and introduce the class of theories that can exhibit a classical limit through decoherence.
In the results section we formally state and outline a proof of our result and in the conclusion section we discuss the physical significance of our result. All technical proofs are given in the corresponding appendices.

\section{Setup}

To begin to pose questions about how different physical features of theories relate we makes use of the generalised probabilistic theories (GPT) framework \cite{hardy2001quantum,barrett2007information,chiribella2010probabilistic,hardy2011reformulating,masanes2006general}, which is broad enough to describe any operational description of nature. The framework is based on the idea that any physical theory must be able to predict the outcomes of experiments, and moreover, that the theory should have an operational description in terms of those experiments. This framework is broad enough to describe arbitrary operationally defined theories including but not limited to quantum and classical theory. We provide a brief introduction to the framework and our notation below. For a full introduction to this framework see \cite{barrett2007information,chiribella2011informational}.

A primitive notion in this framework is the idea of a system $\mathcal{S}$, corresponding for example in quantum theory to an $n$-level quantum system. Such a system can be prepared in a variety of ways and so has an associated set of \emph{states} $\Omega_\mathcal{S}$. One can perform measurements on the system to determine which state it has been prepared in, and the measurement outcomes are known as \emph{effects} $e\in \mathcal{E}_\mathcal{S}$ which are maps $e:\Omega_{\mathcal{S}}\to [0,1]$ determining the probability that outcome $e$ is observed given the system was in state $s$. Moreover, there will generally be \emph{transformations} $T\in\mathcal{T}_\mathcal{S}$ that can be applied to the system, note that if $s\in \Omega_{\mathcal{S}}$ then $T\circ s \in \Omega_{\mathcal{S}}$ and moreover that if $e\in \mathcal{E}_\mathcal{S}$ then $e\circ T \in \mathcal{E}_\mathcal{S}$. Transformations are said to be \emph{reversible} $T\in \mathcal{G}_\mathcal{S}$ if $T^{-1}$ is also a valid transformation where $T\circ T^{-1}=\mathds{1}=T^{-1}\circ T$.

Based on operational ideas we find that these sets of states, effects and transformations have much more structure. Specifically, the state space has the structure of a finite-dimensional convex set. \emph{Convexity} corresponds to the idea that if one can prepare a system in state $s_1$ or $s_2$ then it should be possible to prepare a probabilistic mixture of these two states, for example, conditioned on the outcome of a biased coin flip. If a state can we written as a convex combination of other states $\rho = \sum_i p_i s_i$ we say that $s_i$ \emph{refines} $\rho$ (denoted $s_i \succ \rho$ $\forall$ $i$ \citep{chiribella2010probabilistic}). States that cannot be refined by any other states are called \emph{pure}, otherwise they are called \emph{mixed}. Not all of the well defined measures of purity in quantum and classical theory will translate to general theories, but there is a sufficient condition for if one state is purer than another that applies to all convex theories. In general, if state $\sigma$ can be written as a probabilistic mixture involving $\rho$, but not vice versa (e.g. $\rho \succ \sigma$, $\sigma \not\succ \rho$), then $\rho$ is strictly purer than $\sigma$. Operationally, $\sigma$ can be prepared by an experiment that prepares a probabilistic mixture of states including $\rho$ , but the converse is not true for $\rho$.

Finite dimensionality comes from the requirement that it should be possible to characterise the state of a system by performing only a finite number of distinct experiments. Moreover this state space is typically assumed to be compact and closed. Transformations and effects should respect this convex structure, for example, probabilistically preparing state $s_1$ or $s_2$ followed by applying some transformation $T$ should be operationally equivalent to probabilistically preparing state $T\circ s_1$ or $T\circ s_2$. This implies that transformations and effects should be \emph{linear} maps  \footnote{or can at least with out loss of generality be extended to}.

It is typically useful not to only consider the physical states of the system, but also, to consider sub- and super-normalised states. This extends the state space from a convex set $\Omega_\mathcal{S}$ in a $d$-dimensional vector space to a convex cone $\mathcal{K}_{\mathcal{S}}$ living in a $d+1$ dimensional vector space. The state space $\Omega_\mathcal{S}$ is recovered by enforcing normalization via the deterministic effect $u_\mathcal{S}$. That is that states $s\in \mathcal{K}_\mathcal{S}$ are normalised, and thus belongs to in $\Omega_\mathcal{S}$, if $u_{\mathcal{S}}(s)=1$, sub-normalised if $u_{\mathcal{S}}(s)<1$ and supernormalised otherwise. Effects now extend to linear maps $e:\mathcal{K}_\mathcal{S} \to \mathds{R}^+$ and transformations extend to linear maps on the cone. Reversible transformations therefore must be automorphisms of the cone that preserve the normalised state space. Beyond these minimal assumptions, we place no further constraints on the state space $\Omega_\mathcal S$ which can take the form of arbitrary convex sets (see the supplementary materials for examples). In the statement of our results we use the notion of the \emph{faces} of a convex set. These are defined in the supplementary materials but can be understood intuitively as the convex subsets that form the boundary of the convex set. For example, for the three dimensional cube the faces are the squares, edges and vertices on the boundary of the cube.

The above is best illustrate with examples, the key examples here being quantum theory and classical probability theory. Given an $n$-level quantum system the convex cone is given by the set of positive semi-definite Hermitian matrices and the deterministic effect by the trace, such that normalised states are density matrices. Given an $n$-level classical system the convex cone is given by real vectors with non-negative entries and the deterministic effect by the covector with all entries $1$ such that normalised states correspond to probability distributions over an $n$ element set. Effects are then linear functionals on these cones in quantum theory corresponding to POVM elements and in classical theory to covectors with elements $\leq 1$. Transformations are then linear maps between these cones, in quantum theory corresponding to CP maps and in classical theory to sub-stochastic matrices. Reversible transformations are then cone-automorphisms that preserve the normalised states; in quantum theory these will be unitary transformations and in classical theory permutations of the underlying set.

There is a final key aspect of a theory that we are yet to discuss, and that is how to form \emph{composite systems}. Given two systems $\mathcal{S}_1$ and $\mathcal{S}_2$ with their associated state spaces/cones, effects and transformations, there should be a way to form a composite system, denoted $\mathcal{S}_1\otimes\mathcal{S}_2$. Note that here we use the symbol $\otimes$ to denote the construction of a bipartite system, which need not be related to the vector space tensor product \footnote{Under the assumption of Local Tomography \cite{hardy2001quantum} the vector spaces containing the cones will indeed compose under the standard vector space tensor product, but, there is still freedom to choose how the bipartite cone within this composite vector space \cite{barrett2007information} .}.  There are various operational constraints on this product $\otimes$ \cite{barrett2007information}, for example, that if one can prepare system $\mathcal{S}_1$ in state $s_1$ and $\mathcal{S}_2$ in state $s_2$ then there should be a state, denoted $s_1\otimes s_2$ which represents independently preparing the two systems in these two states. Similar statements and constraints can be made for the effects and transformations of the theory. These operational constraints however do not uniquely specify a way to form composite systems, as generally composite systems allow for more than just doing things on each system independently. Given local state spaces $\mathcal S_1$ and $\mathcal S_2$ there are therefore many different possible composite systems that could be formed (see for example \cite{barrett2007information}). An important feature of these composites is whether or not the bipartite state spaces exhibit entanglement, which we now define.

\begin{definition}[Entanglement]
A state $\psi$ belonging to the bipartite state space $\mathcal S_1\otimes \mathcal S_2$  is entangled iff it cannot be written in the following form
\[\psi = \sum_i p_i s_i \otimes s'_i \quad \sum\limits_i p_i =1, \ \, p_i\geq 0 \, \ \forall\, i \]
where $s_i\in \Omega_{\mathcal S_1}$, $s'_i\in\Omega_{\mathcal S_2}$, i.e. a state is entangled if it cannot be seen as the convex combination of product states.
\end{definition}

In this paper we will show that entangled states are a feature of any non-classical theory which can decohere to classical theory. We prove this by showing that any theory without entanglement that decoheres to classical theory must be classical theory itself. As such, we need a way to define classical theory and theories without entanglement. Constraining a theory to have no entanglement is equivalent to fixing a particular choice of tensor product for the theory, that is, the \emph{min-tensor product} \cite{barrett2007information}. Therefore rather than defining the general requirements of a tensor product we will just consider this particular case.

\begin{definition}[Min-tensor product $\boxtimes$]
The min-tensor product for combining systems $\mathcal A$ and $\mathcal B$ is defined by:
\[\mathcal{K}_{\mathcal{A}\boxtimes\mathcal{B}}:=\mathsf{Conv}\left[\left\{a\otimes b \middle| a\in \mathcal{K}_\mathcal{A}, b\in\mathcal{K}_\mathcal{B}\right\}\right]\]
\[u_{\mathcal{A}\boxtimes\mathcal{B}}=u_\mathcal{A}\otimes u_\mathcal{B}\]
where here $\otimes$ is the vector space tensor product.
\end{definition}

Note that as classical theory exhibits no entanglement, system compose under the minimal tensor product, $\otimes = \boxtimes$. We can now define theories without entanglement and classical theory.

\begin{definition}[Generalised probabilistic theory without entanglement]
A GPT without entanglement is defined by  a collection of systems $\{\mathcal{S}\}$, their associated effects $\mathcal{E}_\mathcal{S}$, and transformations between them $\mathcal{T}_{\mathcal{S}\to\mathcal{S'}}$, and, the composite of systems $\mathcal{S}$ and $\mathcal{S}'$ is given by the min-tensor product, $\mathcal{S}\boxtimes\mathcal{S}'$.
\end{definition}

\begin{definition}[Classical probabilistic theory]\label{def classical}
An $N$-level classical system, denoted $\Delta_N$ has a state space which is an $N$ vertex simplex. These compose under the min-tensor product and satisfy:
\[\Delta_N\boxtimes \Delta_M = \Delta_{NM}\]
Reversible transformations correspond to permutations of the vertices of the simplex. Effects are any linear functional $e : \mathcal K_\mathcal{S} \to \mathds{R}^+, \, e: \Omega_\mathcal{S}\to [0,1]$.
\end{definition}

An interesting feature of classical and quantum theory is that they obey the \emph{no-restriction hypothesis} \cite{chiribella2010probabilistic,janotta2013generalized}, which states all mathematically well defined effects are allowed in the theory and can be experimentally realised. We do not make this assumption when considering theories that can decohere to classical theory. Finally we must define and characterise the generalised decoherence-to-classical map, which we discuss in the following section.

\section{Decoherence}
It is physically well motivated to postulate that, in any reasonable theory of nature, be it quantum or post-quantum, systems must be able to behave classically. Indeed the GPT framework is fundamentally built on the assumption that we have a \emph{classical interface} with the world. We can choose, potentially using classical randomness, which experiment to perform, and we can characterise states, effects and transformations in terms of classical probability distributions that we obtain from experiments. However, ultimately this classical interface should be explainable from the theory itself rather than just being an external structure. This is indeed the case in quantum theory, where we can view the classical interface as an effective description of \emph{decohered} quantum systems. It therefore seems like any well-founded GPT should have an analogous decoherence mechanism so as to explain how it gives rise to the classical interface. We now consider the key features of quantum to classical decoherence which we then take to define decoherence for GPTs. See \cite{gogioso2017categorical,coecke2017twoTEST,selby2017leaks} for closely related approaches to decoherence in generalised theories.

For each quantum system $\mathcal{Q}$ there is a decoherence map $D_\mathcal{Q}$ and classical system with state space $\Delta_{N(\mathcal{Q})}$ where the decoherence map is given by $D_\mathcal{Q}[\rho]:= \sum_{i=1}^N\bra{i}\rho\ket{i}\ket{i}\bra{i}$. This map has the following key properties:

\begin{definition}[Decoherence maps]\label{def dec}
Purely decoherence maps, in quantum and general theories, obey the following properties
\begin{itemize}
\item[1.] Physicality: the decoherence map is a physical map, typically considered to be arising from an interaction with some environmental system that is then discarded, and hence must satisfy all of the constraints on transformations in a GPT. In particular, it must be \emph{linear} and map states to states.
\item[2.] Idempotence: in quantum theory the decoherence map destroys the coherences between the basis states, and so applying it a second time does nothing to the state. In general, the decoherence map should restrict the state space to a classical subspace which is invariant under repeat applications. Therefore, applying it twice is the same as applying it once and $D_\mathcal{S}[D_\mathcal{S}[\sigma]]=D_\mathcal{S}[\sigma]$ $\forall$ $\sigma \in \Omega_\mathcal{S}$.
\item[3.] Purity-decreasing: the decoherence map arises from losing information to an environment and as such it cannot increase our knowledge of the input state. Therefore $D[\rho]$ cannot be strictly purer than $\rho$ for any input state $\rho$. For example in quantum theory a decoherence map will not map mixed states to pure states. In general a state $\rho$ is strictly purer than state $\sigma$ if $\rho \succ \sigma$ and $\sigma \not\succ \rho$. Therefore if $D_\mathcal{S}[\rho]\succ \rho$, then $\rho \succ D_\mathcal{S}[\rho]$, else $D_\mathcal{S}[\rho]$ is strictly purer that $\rho$.
\end{itemize}

\end{definition}
In general this map could decohere to any sub-theory. However, we are interested in particular with decoherence maps that take systems $\mathcal S$ to classical systems.

\begin{definition}[Decoherence to classical theory]\label{def dec clas}
A theory decoheres to classical theory if it has a decoherence map (Definition 5) for each system  which obey the following
\begin{itemize}
\item[1.] State space: the most obvious constraint is that image of the decoherence map is a classical state space: $$ D_{\mathcal S}(\Omega_\mathcal S)=\Delta_{N(\mathcal{S})}$$ However, we don't just want to reproduce the states of classical theory, but the full theory including its dynamical and probabilistic structure.
\item[2.] Effect space: classical effects should also arise from the original theory. That is, that for every effect in classical theory there is some effect in the full theory that behaves as the classical effect when we restrict to $\Delta_{N(\mathcal{\mathcal S})}$. This can be formalised as for all $ \ e_\text{classical}$ there exists $e\in \mathcal E_\mathcal{S}$ such that $$ \ e_\text{classical}=e\circ D_{\mathcal{S}}$$
\item[3.] (Reversible) Transformations: similarly for (reversible) transformations we expect for any classical (reversible) transformation $t$ there is a corresponding post-classical (reversible) transformation with the same action on the image of $D_\mathcal{S}$. This can be formalised as for all classical reversible transformations $T_\text{classical}$ there exists some reversible $T\in \mathcal T_\mathcal{S}$ such that $$T_\text{classical} = T \circ D_\mathcal{S} $$
\item[4.] Composites: finally, we expect decoherence to act suitably with composition, i.e. if system $\mathcal{S}_1$ decoheres to $\Delta_N$ via $D_{\mathcal{S}_1}$ and system $\mathcal{S}_2$ to $\Delta_M$ via $D_{\mathcal{S}_2}$, then the composite system $\mathcal{S}_1\otimes\mathcal{S}_2$ can decohere to $\Delta_N\otimes\Delta_M=\Delta_{NM}$ via $D_{\mathcal{S}_1}\otimes D_{\mathcal{S}_2}$.
\end{itemize}
\end{definition}

To give an example consider quantum theory for qubits, which decohere to classical theory for bits by applying the standard dephasing map in (for example) the computational basis $\{\left| 0 \right> \!\! \left< 0 \right|,\left| 1 \right> \!\! \left< 1 \right|  \}$. Measurements in the computational basis provide the classical measurements, utilising the quantum effects $\left| 0 \right> \!\! \left< 0 \right|$ and $\left| 1 \right> \!\! \left< 1 \right|$. Two qubits can decohere independently to classical bits, which can then interact and all permutations of composite bits can be achieved by the classical c-not and bit flip operations, which are provided by the quantum c-not unitary and the Pauli-X unitary. Also note that the dephasing map obeys all conditions in Definition \ref{def dec}.

\section{Results}
We are now in a position to prove our main result. If a theory can decohere to classical theory and does not have entanglement, then the original systems must be composites including a classical system, $\Omega_\mathcal{S}=\Delta_N \otimes \Omega_f$, and the decoherence map simply discards any non-classical subsystems. More succinctly,

\begin{quote}\begin{centering}
\emph{Theories with non-trivial decoherence\\ must have entangled states.\\}
\end{centering}\end{quote}

The proof of this is provided in the appendix along with all necessary mathematical definitions and background to understand the proof. However, we will also provide an outline of the proof here. We first show -- by considering the consequences of decoherence for single systems -- that the state space $\Omega_\mathcal{S}$ has the following geometric properties

\noindent\textbf{Result 1:} \emph{If $D_\mathcal{S}[\Omega_{\mathcal S}]=\Delta_N$ where $D_\mathcal{S}$ obeys Definitions \ref{def dec} and \ref{def dec clas} then the state space $\Omega_\mathcal{S}$ has the following properties}
\begin{itemize}
\item[1.] $\Omega_\mathcal{S}$ is the minimal face of a set of faces $f_1, f_2, \dots , f_N$ that are isomorphic $f_i\cong f_j$ $\forall$ $i,j$, disjoint (share no states) $f_i\cap f_j =\emptyset$ $\forall$ $i,j$ and are exposed
\item[2.] each face $f_i$ decoheres uniquely to a pure classical state $s_i$, $D_\mathcal{S}[f_i]=s_i$.
\end{itemize}

We then consider the consequences of decoherence on composite systems. Essentially, as the resulting classical systems must be able to interact under classical dynamics we deduce the following additional constraints

\noindent\textbf{Result 2:} \emph{If $D_{\mathcal{S}}[\Omega_\mathcal S]=\Delta_N$ where $D_{\mathcal{S}}$ obeys Definitions \ref{def dec} and \ref{def dec clas} and $D_{\mathcal{S}}[\Omega_\mathcal S]\boxtimes D_{\mathcal{S}}[\Omega_\mathcal S]=\Delta_{N^2}$ which enjoys the full set of classical dynamics in Definition \ref{def classical}, then}
\begin{itemize}
\item[1.] The classical faces are linearly independent, $\mathsf{Span}[f_i] \cap \mathsf{Span}[\bigcup\limits_{k\neq i} f_k] = \{ 0 \}\ \ \forall i$
\item[2.] $\Omega_\mathcal{S}$ is the convex hull of these classical faces $\Omega_\mathcal{S}=\bigvee_i f_i = \mathsf{Conv}[\{f_i\}]$
\end{itemize}
Given these constraints on the state space it is then simple to show that it is a composite of a classical state space with a state space isomorphic to these faces, and the decoherence map simply discards the non-classical subsystem

\noindent\textbf{Result 3:} \emph{For any theory that decoheres to classical theory as per Definitions \ref{def dec} and \ref{def dec clas}, and whose composition rule $\otimes$ is given by the minimal tensor product $\boxtimes$, all state spaces are of the form }
\[\Omega_S =  f \boxtimes\Delta_N\]
\noindent \emph{and the decoherence map is of the form }
\[D_\mathcal{S} = (s\circ u)_f\otimes \mathds 1_{\Delta_N} \]
\noindent\emph{where $u$ is the discarding effect and $s$ is some fixed internal state of $f$. E.g. decoherence of non-classical systems comprises of discarding them.}

Therefore, if we restrict ourselves to considering non-classical theories, the only decoherence-to-classical map possible is the trivial map where we discard our non-classical systems. For example, in quantum theory this would correspond to all quantum systems regardless of their state or dimension decohereing to the zero-dimensional classical state, and the resulting classical theory being trivial.

\section{Discussion}
In this article we have shown that if a theory has a non-trivial decoherence mechanism, such that decoherence isn't simply discarding the system, then the theory must have entangled states. It therefore seems that entanglement, rather than being a surprising feature of nature, is an entirely inevitable feature of any post-classical theory. A natural question to ask is what other features of quantum theory can be reproduced simply by demanding that the theory has a classical limit?

There are myriad other physical features that could be implied from the existence of a classical limit such as information causality \cite{pawlowski2009information}, bit symmetry \cite{muller2012structure} and macroscopic locality \cite{navascues2009glance} to name but a few. Of particular interest would be deriving genuine device-independent non-locality. The existence of entangled states is in general a necessary but insufficient condition for observing violations of Bell inequalities. For example, non-separable states are present in the local theory of Spekken's toy model \cite{spekkens2007evidence}. On the other hand, it has been shown that all entangled states in quantum theory display some hidden non-locality \cite{buscemi2012all,masanes2008all}. By determining the additional structure present in quantum theory that gives this correspondence between entanglement and non-locality, it could be possible to derive the violation of Bell inequalities from purely physical postulates. Given the simplicity of the postulates used to derive the existence of entangled states, it is plausible that the postulates that give rise to Bell non-locality are similarly mundane.

This notion of decoherence has allowed us to define an interesting class of GPTs -- those with a classical limit \footnote{Theories with a classical limit also have been considered in earlier work within the closely related frameworks of Categorical Probabilistic Theories \cite{gogioso2017categorical} and process-theories \cite{selby2017leaks,coecke2017twoTEST}.}. There is clear physical motivation to consider this class. For example, if a theory were not to have such a limit then one would have to posit the existence of two fundamentally distinct types of systems, the classical systems (which are how we interact with the world) along with post-classical systems (which we cannot directly probe). Such a fundamental distinction appears unnatural, and so it seems that decoherence is a necessary feature of any sensible operational theory. However, whilst being a physically well-motivated class, it nonetheless provides a great deal of mathematical structure and as such gives a more powerful framework for studying generalised theories.

\noindent \textbf{Acknowledgments:} The authors would like to thank the reviewers for this paper whose insightful criticism led to the development of the ideas presented in this current version. Additionally, the authors would like to thank Llu\'is Masanes and Ciar\'{a}n Lee for useful discussion and Fabio Costa and Magdalena Zych for pointing out a flaw in the proof in the previous version of this paper. JR and JHS are supported by EPSRC through the Controlled Quantum Dynamics Centre for Doctoral Training.
\bibliographystyle{plain}
\bibliography{bibliography}

\appendix

\widetext

In these appendices we first provide sufficient mathematical background to understand the result and secondly, provide the proof of our main theorem.

\section{Mathematical background}\label{app-mathsbackground}

As mentioned in the main body, we associate each system $\mathcal{S}$ with a state space $\Omega_\mathcal{S}$ which is defined as
a finite dimensional compact closed convex set.  That is to say, a closed and bounded set of vectors in a real, finite-dimensional vector space such that if $a_1$ and $a_2$ are inside the set then for $p\in[0,1]$, $pa_1+(1-p)a_2$ is also in the set.

It will often be useful to work with the convex pointed cone $\mathcal K_\mathcal{S}$ generated by convex set $\Omega_\mathcal{S}$, for which we use the standard definition \cite{barker1981theory}, for a more detailed background on convex sets and cones see e.g. \cite{valentine1964convex}.
\begin{definition}
$\mathcal K_\mathcal{S}$ is the cone generated by a convex set $\Omega_\mathcal{S}\subseteq \mathds{R}^d$, defined as:
 \[\mathcal{K}_\mathcal{S} := \mathsf{Conv}\left[\left\{ \lambda (s,1) \middle| \lambda\in \mathds{R}^+, s\in \Omega_\mathcal{S}\right\}\right]\subseteq \mathds{R}^{d+1}.\]
\end{definition}

An important feature of convex sets is that they have a partial ordering in terms of \emph{refinements} \citep{chiribella2010probabilistic}.
\begin{definition}
We say that a state $s$ refines a state $s'$, denoted $s\succ s'$, if there is some convex decomposition of $s'$ that includes $s$, more formally,
\[s\succ s' \ \ \iff \ \ \exists p\in [0,1], t \text{ such that } s'=ps+(1-p)t.\]
\end{definition}
Operationally this means that $s'$ is `less pure' than $s$ as it can be written as a probabilistic mixture which includes the state $s$. This partial ordering on states allows us to define another important concept, that is, \emph{faces} of the state space.
\begin{definition}
A face $f$ is a subset of the state space which is closed under refinements, that is, if a state refines a state in the face then it must also be in the faces, formally:
\[s' \in f \text{ and } s\succ s' \quad \implies s \in f\]
\end{definition}
The 0 dimensional faces are referred to as \emph{vertexes}, denoted $\mathsf{Vertex}(\mathcal S)$, and correspond to pure states. Faces inherit an partial ordering from the ordering of states, that is, $f\leq g \iff \exists s_f \in f, s_g \in f \text{ s.t. } s_f\succ s_g$. This ordering can be shown to have the structure of a \emph{lattice}. Importantly, this means that the set of faces have two binary operations, $\land$ and $\lor$. $f\land g$ corresponds to the intersection of $f$ and $g$ whilst $f\lor g$ is the minimal face that contains the faces $f$ and $g$. Furthermore, faces are \emph{exposed} iff there is a supporting hyperplane that intersects the convex set at the face only, otherwise the face is not exposed.

\begin{center}
\begin{figure}[h!]
\includegraphics[scale=0.7]{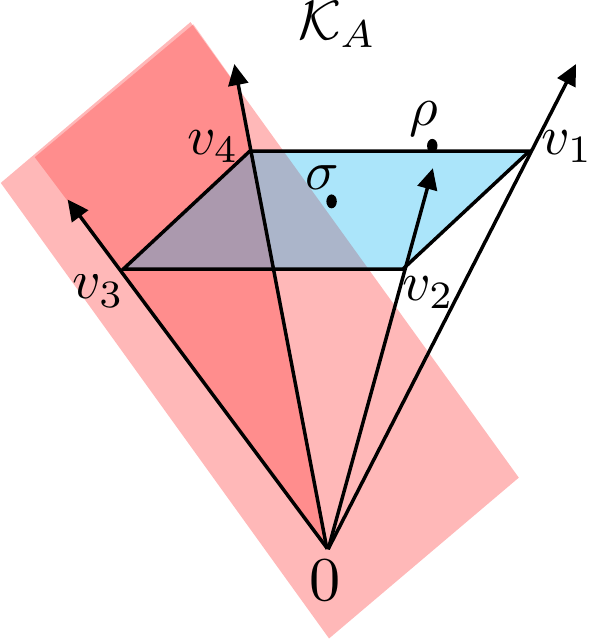}
\caption{Figure shows the cone $\mathcal K_A$ generated by state space $A$ which is a square. The four pure states correspond to the four extremal points $v_i\in\mathsf{Vertex}(A)$ and the arrows correspond to the extremal rays of the cone, which are generated by these points. The point of the cone is shown as $0$. $\vee \rho = \mathsf{conv}(v_1, v_4)$, and $v_1, v_4 \succ \rho$. $\vee \sigma = A$ and $\rho \succ \sigma$. The red plane depicts a supporting hyperplane of the cone, which supports the face of the cone generated by the face of the state space $\mathsf{conv}(v_3, v_4)$ which in turn spans this hyperplane. }\label{fig: cone}
\end{figure}
\end{center}

\begin{center}
\begin{figure}[h!]
\includegraphics[scale=0.5]{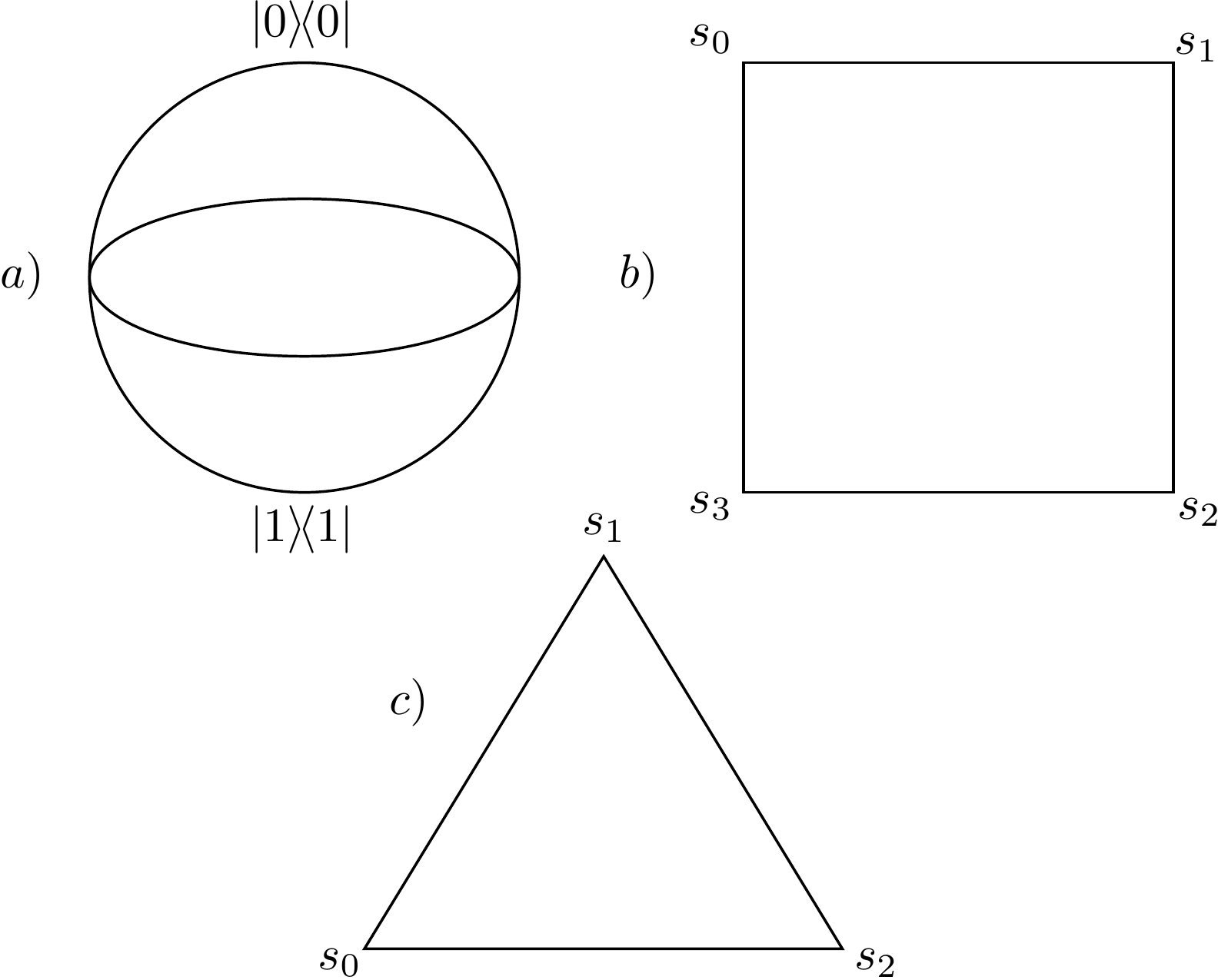}
\caption{Figure shows the local state spaces of various GPTs including $a)$ a qubit, $b)$ a Popescu-Rohrlich box \cite{popescu1994quantum} and $c)$ a classical 3-level system (trit). The latter two have a finite number of pure states which are indicated.  }\label{figure: convex sets}
\end{figure}
\end{center}

 Given two systems $\mathcal A$ and $\mathcal B$ there are two constructions of composite state cones which are important for the derivation of our results.

\begin{definition}Minimal tensor product, $\mathcal{A}\boxtimes \mathcal{B}$:
\[\mathcal{K_A}\boxtimes \mathcal{K_B}:=\mathsf{Conv}\left[\left\{a\otimes b \middle| a\in \mathcal{K_A}, b\in \mathcal{K_B} \right\} \right]\]
where $\otimes$ is the vector space tensor product such that $\mathcal{K_A}\boxtimes \mathcal{K_B}\subseteq \mathsf{Span}[\mathcal{K_A}]\otimes \mathsf{Span}[\mathcal{K_B}]$.
\end{definition}

\begin{definition}Direct sum, $\mathcal{A}\oplus \mathcal{B}$:
\[\mathcal{K_A}\oplus \mathcal{K_B}:=\mathsf{Conv}\left[\left\{a\oplus {\bf 0}, {\bf 0}\oplus b \middle| a\in \mathcal{K_A}, b\in \mathcal{K_B} \right\} \right]\]
where $\oplus$ on the RHS is the vector space direct sum such that $\mathcal{K_A}\oplus \mathcal{K_B}\subseteq \mathsf{Span}[\mathcal{K_A}]\oplus \mathsf{Span}[\mathcal{K_B}]$.
\end{definition}

There is an important relation between these two constructions, namely that $\boxtimes$ distributes over $\oplus$.
\begin{lemma}\label{lem: distribute}
The min-tensor $\boxtimes$ distributes over the direct sum $\oplus$.
\end{lemma}
\begin{proof}
Consider three state spaces $A$, $ B$ and $ C$. We want to show that,
\[ A\boxtimes (B\oplus C) = (A\boxtimes B)\oplus(A\boxtimes C)\]
This follows immediately from writing out both sides using the definitions of $\boxtimes$ and $\oplus$, using distributivity of the direct sum and tensor product on the individual vectors, and noting that ${\bf 0} \otimes s = {\bf 0}\otimes {\bf 0}$.
\end{proof}

We can characterise whether a state space is decomposable over the direct sum by observing certain properties of the cone.
\begin{definition}\label{def: decomp}
 A state space $\mathcal{S}$ is decomposable $\mathcal{K_S}=f\oplus g$ if and only if there are a pair of faces $f,g \in \mathsf{Face}[\mathcal{K_S}]$ such that: i) the state space is the convex hull of the two faces, $\mathcal{K_S}=f\lor g = \mathsf{Conv}[\{f,g\}]$; ii) the faces are linearly independent, $\mathsf{Span}[f]\cap \mathsf{Span}[g]=\mathbf{0}$.
\end{definition}

\

\begin{lemma}\label{lem:LatMorph2} $T(f)=f'\cong f$ i.e. faces are mapped to isomorphic faces. \end{lemma}
\proof Consider some $a'\in f'$ then $f'$ is a face if for any decomposition, $a'=p a'_1 +(1-p)a'_2$, $a'_i$ are also in the set $f'$. Note that $a'=T(a)$ for some $a\in f$ and as $T$ is reversible, this means that $a=T^{-1}(a')=T^{-1}(p a'_1 +(1-p)a'_2)$ which by linearity of $T^{-1}$ implies that $a=pT^{-1}(a'_1)+(1-p)T^{-1}(a'_2)$. This provides a decomposition of $a$ and as $f$ is a face this means that $T^{-1}(a'_i)\in f$ and so $a'_i$ are both in $f'$. Therefore $f'$ is a face. It is clearly isomorphic to $f$ as $T$ provides the isomorphism. \endproof

\begin{lemma}\label{lem:LatAutoMorph}A reversible transformation $T$ on $A$ induces an automorphism of the face lattice $\mathsf{Face}(A)$.
\end{lemma}
\proof Lemma \ref{lem:LatMorph2} shows that faces are mapped to faces, then reversibility of $T$ implies that this mapping of faces must be 1 to 1 and hence induces a lattice automorphism.\endproof
An immediate corollary of this is that:
\begin{corollary}\label{lem:LatMorph1}
$T(a_1\lor a_2) =T(a_1)\lor T(a_2)$
\end{corollary}

\subsection{Single system consequences of decoherence}
We begin by considering a state space $\Omega_\mathcal{S}$ with a decoherence map $D_\mathcal{S}$ such that $D_\mathcal{S}[\Omega_\mathcal{S}]=\Delta_N$. We label the vertices of this classical state space as $s_i\in \mathsf{Vertex}[\Delta_N]$ for $i=1,...,N$.

\

\begin{lemma}\label{lem:RefinmentSetFaces} Every classical pure state $s_i$ is contained in a face $f_i$ that is the refining set of $s_i$. \end{lemma}
\proof A face $F$ is defined as a subset of the convex set such that if $s \in F$ and $s= p s_1 + (1-p) s_2$ then $s_i\in F$. It is simple to show that a refining set satisfies this property.
\endproof

\begin{lemma}\label{lem:FaceDecoheres}
Every state in $f_i$ decoheres to $s_i$ and every state that decoheres to $s_i$ is in $f_i$: $D_\mathcal{S}[s]=s_i \iff s_i \in f_i$.
\end{lemma}
\proof
Consider some state $s\in f_i$, this means that $s$ is in the refining set of $s_i$ that is, $s\succ s_i$ and $s_i = \sum_\alpha p_\alpha \sigma_\alpha$ where $\sigma_0 = s$. Then linearity and idempotence of $D_\mathcal{S}$, $D_\mathcal{S}[s_i]=\sum_\alpha p_\alpha D_\mathcal{S}[\sigma_\alpha]$ and therefore $D_\mathcal{S}[s]\succ D_\mathcal{S}[s_i] = s_i$. However $s_i$ is, by definition, a vertex of the image of $D_\mathcal{S}$ which implies that $D_\mathcal{S}[s]=s_i$. Conversely, assume that $D[s]=s_i$, and for the sake of contradiction assume that $s\not\in f_i$ which implies that $s\not\succ s_i = D_\mathcal{S}[s]$. Define a state $\tau := p s_i + (1-p) s$. Clearly $s_i\succ \tau$ but $\tau\not\succ s_i$. However $D_\mathcal{S}[\tau]=s_i$ and therefore $\tau\not\succ D_\mathcal{S}[\tau ]$, $D_\mathcal{S}[\tau ] \succ \tau$. Thus $D_\mathcal{S}[ \tau ]$ is strictly purer that $\tau$ violating the purity-decreasing postulate 3 of Definition 5.
\endproof

\begin{lemma}
The $f_i$ corresponding to distinct pure classical states are disjoint: $f_i\cap f_j = \emptyset \ \ \forall \ i\neq j$.
\end{lemma}
\proof
By Lem.~\ref{lem:FaceDecoheres} we know that $D_\mathcal{S}[f_i]=s_i$ and so $f_i$ is contained within the preimage of $s_i$. Hence, if the $s_i$ and $s_j$ are distinct then their preimages are disjoint and so $f_i\cap f_j =\emptyset$.
\endproof

\begin{lemma}\label{lem:isomorphicFaces}
The $f_i$ are all isomorphic: $f_i\cong f_j \ \ \forall\ i,j$.
\end{lemma}
\proof
In classical theory there is a reversible transformation $\pi:: s_i \mapsto s_j$, hence, by Postulate 3 of Definition 6 there is a reversible transformation in the post-classical theory $T_{\pi}::s_i\mapsto s_j$. As $T_\pi$ is linear it must map the refining set of $s_i$ into the refining set of $s_j$ and, as $T_{\pi}$ is reversible this defines an isomorphism between $f_i$ and $f_j$.
\endproof

\begin{lemma}\label{lem:ExposedFaces1}
The $f_i$ are exposed faces: there is a linear functional $E_i$ such that $E_i(s)\geq 0\ \forall s\in\mathcal{S}$ and $E_i(s)=0 \iff s\in f_i$.
\end{lemma}
\proof
In classical theory we have effects (i.e. linear functionals that are non-negative on the cone) $E_i$ such that $E_i(s)\geq 0 \ \ \forall s\in \Delta_N$ and $E_i(s)=0 \iff s=s_i$. By Postulate 2 of Definition 6 these arise from the post-classical theory as $E_i=\epsilon_i \circ D_\mathcal{S}$ and so by idempotence of decoherence we have $E_i\circ D_\mathcal{S}=E_i$. Therefore, $E_i(s)=E_i(D_\mathcal{S}[s]) \geq 0 \ \ \forall s \in \mathcal{S}$ and $E_i(s)=E_i(D_\mathcal{S}[s])= 0 \iff D_\mathcal{S}[s]=s_i \iff s\in f_i$.
\endproof

{\begin{lemma}\label{lem:ExposedFaces}
The faces $f_{I}:=\bigvee_{i\in I} f_i$ are all exposed faces.
\end{lemma}
\proof
First note that Lem.~\ref{lem:RefinmentSetFaces} implies that the refinement set of any state is a face. Moreover, note that if a state belongs to a face, then the refinement set of that state is a subset of that face. Also note that we always have $\mathsf{Conv}[\{f_i\}_{i\in I}]\subseteq\bigvee_{i\in I} f_i$. Which, as $s_i\in f_i$, implies that $\frac{1}{|I|}\sum_{i\in I} s_i \in \bigvee_{i\in I}f_i$. Hence the refinement set of $\frac{1}{|I|}\sum_{i\in I} s_i$, denoted $\mathsf{Ref}\left[\frac{1}{|I|}\sum_{i\in I} s_i\right]$ satisfies $\mathsf{Ref}\left[\frac{1}{|I|}\sum_{i\in I} s_i\right]\subseteq \bigvee_{i\in I}f_i$. However, it is also the case that for all $i\in I$ we have $f_i\subseteq \mathsf{Ref}\left[\frac{1}{|I|}\sum_{i\in I} s_i\right]$ hence -- by minimality of $\lor$ and the fact that the refinement set is a face -- we must have $\bigvee_{i\in I} f_i \subseteq \mathsf{Ref}\left[\frac{1}{|I|}\sum_{i\in I} s_i\right]$. Combining these gives us that \[\mathsf{Ref}\left[\frac{1}{|I|}\sum_{i\in I} s_i\right] = \bigvee_{i\in I} f_i := f_I\]

Now we can consider a generalisation of Lem.~\ref{lem:FaceDecoheres}. It is simple to see that $D[f_I] = \textsf{Conv}[\{s_i\}_{i\in I}]$, from idempotence of $D$ and the fact that $s_j\in f_I \iff j \in I$. We then want to show that
\[D[s] \in \textsf{Conv}[\{s_i\}_{i\in I}] \implies s\in f_I\]
Assume the converse for the sake of contradiction, that is, there is a state $\sigma\not \in f_I$ such that $D[\sigma]\in\textsf{Conv}[\{s_i\}_{i\in I}]$. We can then define a state $\tau := p D[\sigma] +(1-p) \sigma$ such that $D[\sigma]\succ\tau$ but as $D[\sigma]\in f_I$ and $\sigma\not\in f_I$ it is clear that $\tau\not\succ D[\sigma]$. Now note that $D[\tau] = D[\sigma]$ so combining these we find that $D[\tau]\succ \tau$ but $\tau\not\succ D[\tau]$ which contradicts the purity-decreasing postulate 3 of Definition 5. Therefore it must be that $\sigma \in f_I$.

Finally to complete this proof we generalise Lem.~\ref{lem:ExposedFaces1}. We can consider the classical effect $E_I$ defined such that $E_I(s_j) = 0 \iff j\in I$. And note that this must arise from an effect $\epsilon_I$ in the full theory as $E_I:=\epsilon_I\circ D$. It is then simple to see, from idempotence of $D$, that $E_I(s)=0 \implies E_I(D[s])=0$ which means that $D[s]\in \textsf{Conv}[\{s_i\}_{i\in I}]$ and so, by the above result, $s\in f_I$. This concludes the proof as $E_I$ is an effect that exposes the face $f_I$.
\endproof
}

\begin{lemma}
The minimal face that contains $\Delta_N$ is the entire state space: $\bigvee_{i=1}^N f_i = \mathcal{S}$.
\end{lemma}
\proof
Consider some state in the interior of the state space i.e. $\rho \in \mathsf{Int}[\mathcal{S}]$. This means that every other state refines $\rho$, i.e. $s\succ \rho \ \ \forall \ s \in \mathcal{S}$. Now consider $D_\mathcal{S}[\rho]$, and note that if a state is refined by an interior state then it too is an interior state. Therefore, for $D_\mathcal{S}[\rho]$ to not be an interior state it must be that $\rho \not\succ D_\mathcal{S}[\rho]$ but $D_\mathcal{S}[\rho]\succ \rho$, this therefore violates the purity-decreasing Postulate 3 of Definition 5. Therefore, $D_\mathcal{S}[\rho]$ must also be an interior state. As $D_\mathcal{S}[\rho] \in \Delta_N$ is interior to $\mathcal{S}$ then the minimal face that contains $\Delta_N$ must be the entire state space.
\endproof

\begin{result}
To summarise we have, from considering decoherence for single systems found the following properties of a state space $\mathcal{S}$ that decoheres via $D_\mathcal{S}$ to $\Delta_N$.
\begin{enumerate}
\item Every classical pure state $s_i$ is contained in a face that is its refining set $f_i$.
\item These faces decohere to their associated pure classical state and any state that decoheres to the classical state is in the face: $D_\mathcal{S}[s]=s_i \iff s\in f_i$.
\item The faces corresponding to distinct classical pure states are disjoint: $f_i\cap f_j=\emptyset\ \ \forall i\neq j$.
\item These faces are all isomorphic: $f_i\cong f_j \ \ \forall i,j$.
\item These faces are `exposed': there exists an effect $E_i$ such that $E_i(s)=0 \iff s\in f_i$.
\item The minimal face containing all of these $f_i$ is the whole state space: $\bigvee_i f_i = \mathcal{S}$
\end{enumerate}
\end{result}

\subsection{Bipartite system consequences of decoherence}

Now let us turn our attention to bipartite systems, in particular, let us consider the bipartite state space $\mathcal{S}\boxtimes\mathcal{S}$ which decoheres via $D_\mathcal{S}\boxtimes D_\mathcal{S}$ to $\Delta_N\boxtimes\Delta_N=\Delta_{N^2}$. For example, one can consider these two state spaces decohereing independently and then being brought together.

\begin{lemma}\label{lem:ProductRefiningSets}
The refining set for $s_i\otimes s_j$ is $f_i \otimes f_j$.
\end{lemma}
\proof
In the minimal tensor product a generic state can be written as $\sum_\alpha p_\alpha a_\alpha \otimes b_\alpha$. Therefore, if $s_i\otimes s_j = \sum_\alpha p_\alpha a_\alpha \otimes b_\alpha$ then by considering the two marginals by applying
 $u \otimes \mathds 1$ and $\mathds 1 \otimes u$ (i.e discarding each of the sub-systems in turn) we get that $s_i = \sum_\alpha p_\alpha a_\alpha$ and $s_j = \sum_\alpha p_\alpha b_\alpha$, and so $a_\alpha \in f_i$ and $b_\alpha \in f_j$ hence the refining set for $s_i\otimes s_j$ is just $f_i\otimes f_j$.
\endproof

\begin{lemma}\label{lem:BipartiteTransformation} There exists a reversible transformation $T$ such that \[T::\begin{cases}f_k\otimes f_i \mapsto f_k\otimes f_i \quad \forall \, i\neq j \\ f_k\otimes f_j \mapsto f_j\otimes f_j\end{cases}\]
\end{lemma}
\proof For a classical bipartite system $\Delta_N\boxtimes \Delta_N$ there is a reversible transformation achieving the permutation of classical states
\[\pi::\begin{cases}s_k\otimes s_i \mapsto s_k\otimes s_i \quad \forall \, i\neq j \\ s_k\otimes s_j \mapsto s_j\otimes s_j\end{cases}\]
therefore by Postulate 3 of Definition 6 there must be a reversible transformation $T_{\pi}$ which has the same action as $\pi$ on the image of $D_\mathcal{S}$. By the same argument as in the proof of Lem.~\ref{lem:isomorphicFaces} this must extend to the relevent refining sets as given by Lem.~\ref{lem:ProductRefiningSets}, hence taking $T=T_{\pi}$ we obtain our result.
\endproof

\begin{lemma} The minimal face containing a set of faces $\{f_i\}_{i=1}^n$ is the convex hull of these faces, $\bigvee_{i=1}^n f_i = \mathsf{Conv}[\{f_i\}_{i=1}^n]$
\end{lemma}
\proof
We prove this result via induction, the $n=1$ case is trivial as it simply states that $f_i=\mathsf{Conv}[\{f_i\}]$. Now for the induction, assume that $\tilde{f}:=\bigvee_{i=1}^{n-1}f_i = \mathsf{Conv}[\{f_i\}_{i=1}^{n-1}]$ then {Lem.~\ref{lem:ExposedFaces} implies this face is exposed, we denote the effect that exposes the face as $\tilde{E}$, and note that this means that  $\tilde{E}(s)=0 \iff s \in \tilde{f}$}. Now define the effect
\[e:=E_n\otimes \tilde{E} + E_1 \otimes E_n\]
and note that $e(a\otimes b)= 0 \iff a\otimes b \in \{\tilde{f}\otimes f_1,f_n\otimes f_n\}$. Hence, $\mathsf{Conv}[\{f_1\otimes \tilde{f},f_n\otimes f_n\}]$ is a face and so we have
\[f_i\otimes \tilde{f} \lor f_n\otimes f_n = \mathsf{Conv}[\{\tilde{f}\otimes f_1,f_n\otimes f_n\}].\]

Now consider the transformation $T$ such that,
\[T::\left\{\begin{array}{lc}f_1\otimes f_i \mapsto f_i\otimes f_i & i=1,...,n-1 \\ f_1\otimes f_n \mapsto f_n\otimes f_n & \end{array}\right.\]
which exists by Lem.~\ref{lem:BipartiteTransformation}. Therefore we have
\[\begin{array}{rl}f_1\otimes (\tilde{f}\lor f_n) & \stackrel{T}{\cong}f_i\otimes \tilde{f}_i \lor f_n\otimes f_n\\ & = \mathsf{Conv}[\{f_1\otimes\tilde{f},f_n\otimes f_n\}]\\ & \stackrel{T}{\cong} f_1\otimes \mathsf{Conv}[\{\tilde{f},f_n\}]\end{array}\]
Hence we have \[\bigvee_{i=1}^n f_i = \tilde{f}\lor f_n = \mathsf{Conv}[\{\tilde{f},f_n\}]= \mathsf{Conv}[\{f_i\}_{i=1}^n]\] as we required.
\endproof

\begin{lemma}\label{lemma: non identical spans}
If the spans of $f_i$ and $f_j$ are equal then they must be the same face: $\mathsf{Span}[f_i]\cap \mathsf{Span}[f_j]=\mathsf{Span}[f_i]\ \iff\  i=j$.
\end{lemma}
\proof
Firstly note that as $f_i\cong f_j$ the dimension of the spans are the same, hence this is equivalent to the statement that
\[\mathsf{Span}[f_i]=\mathsf{Span}[f_j] \ \ \iff \ \ i=j.\]
The $\Leftarrow$ direction is trivial so we focus on the $\Rightarrow$ direction. Consider some state $\chi \in f_j$, by assumption we have $\chi \in \mathsf{Span}[f_i]$ and so $\chi = \sum_\alpha x_\alpha \chi_\alpha$ where $x_\alpha \in \mathds{R}$ and $\chi_\alpha \in f_i$. Therefore $E_i(\chi)=\sum_\alpha x_\alpha E_i(\chi_\alpha)=0$ and so $\chi \in f_j$. This is true for any state of $f_i$ and an identical proof holds with $i$ and $j$ swapped, therefore $f_i=f_j$.
\endproof

\begin{lemma} $\mathsf{Span}[f_i] \cap \mathsf{Span}[\bigcup\limits_{k\neq i} f_k] = \{ 0 \}\ \ \forall i$
\end{lemma}
\proof
First note that given vector spaces $U$ and $V$ with subspaces $A,B \subseteq U$ and $C,D \subseteq V$ that
\[(A\otimes C)\cap (B\otimes D)=(A\cap B)\otimes (C\cap D).\] (see for example Lemma 1.4.5 of \cite{dascalescu2000hopf}). Now consider the transformation $T$ from Lem.~\ref{lem:BipartiteTransformation}.\[T::\begin{cases} f_1\otimes f_i \mapsto f_i\otimes f_i \quad 1\neq i  \\ f_1\otimes f_k \mapsto f_1\otimes f_k \quad \forall \, k\neq i\end{cases}\]
Define $A=\mathsf{Span}[f_1]$, $B:=\mathsf{Span}[f_i]$ and $C:=\mathsf{Span}[\bigcup_{k\neq i} f_k]$. As $T$ is a linear reversible transformation it preserves the intersection of these subspaces. Equating the intersection of the subspaces in the pre-image and the image under $T$ gives
\[\begin{array}{rl}A\otimes (B\cap C) & \stackrel{T}{\cong} (B\otimes B)\cap (A\otimes C)\\  & = (A\cap B)\otimes (B\cap C)\end{array}\]
This can only be satisfied if $A\cap B=A$ or $B\cap C =\{ 0 \}$. The first case is ruled out by Lemma \ref{lemma: non identical spans} whereas the second case implies that $\mathsf{Span}[f_i]\cap \mathsf{Span}[\bigcup _{k\neq i} f_k]=\{ 0 \}$ as we require.
\endproof

\begin{result}
To summarise, from considering decoherence for bipartite systems $\mathcal{S}\boxtimes\mathcal{S}$ which decohere via $D_\mathcal{S}\boxtimes D_\mathcal{S}$ to $\Delta_N\boxtimes\Delta_N=\Delta_{N^2}$ we have derived the following properties of the theory.
\begin{enumerate}
\item The classical faces in Result 1 $f_i$ have linearly independent spans  $\mathsf{Span}[f_i] \cap \mathsf{Span}[\bigcup\limits_{k\neq i} f_k] =\{0\}\ \ \forall i$
\item The minimal face containing a set of faces $\{f_i\}$ is just the convex hull of these faces, $\bigvee_i f_i = \mathsf{Conv}[\{f_i\}]$
\end{enumerate}
\end{result}

\subsection{Proof of main theorem}
We can now combine Results 1 \& 2 to prove the main theorem, restated here for convenience:

\

\noindent \emph{Theorem 1}:
If a theory decoheres to classical theory (Def. 6) then each state space $\Omega_\mathcal{S}$ -- which decoheres via $D_\mathcal{S}$ to $\Delta_{N(\mathcal{S})}$ -- is of the form $\mathcal{A}_{\mathcal{S}}\boxtimes\Delta_{N(\mathcal{S})}$ and the decoherence map is $D_\mathcal{S}=(x_\mathcal{S}\circ u_\mathcal{S})_{\mathcal{A}_{\mathcal{S}}}\boxtimes \mathds{1}_{\Delta_{N(\mathcal{S})}}$, where $x_\mathcal{S}$ is some fixed state of the system and $u_\mathcal{S}$ is the discarding map for the system.

\
\proof
Given the previous results we have shown that the state space can be written as $\Omega_\mathcal{S}=\bigvee_{i=1}^N f_i$ where the $f_i$ are linearly independent. Moreover, $\bigvee_{i=1}^n f_i = \mathsf{Conv}[\{f_i\}_{i=1}^n]$. The conjunction of these conditions implies that the cone generated by $\mathcal S$ decomposes into the direct sum \citep{barker1981theory}
\[\Omega_\mathcal{S}=\bigoplus_{i=1}^N f_i\]
Moreover, as the $f_i$ are all isomorphic we can define $\mathcal{A}\cong f_i \ \ \forall \ i$. Consider the state space \[ \mathcal S' = \mathcal A \boxtimes \Delta_N \] Using $\Delta_M = \bigoplus\limits_{i=1}^N p_i$ where $p_i$ is a zero-dimensional face, and distributivity of $\boxtimes$ over $\oplus$ (Lemma \ref{lem: distribute}) we can write
\[\Omega_\mathcal{S'}=\bigoplus\limits_{i=1}^N \mathcal A \boxtimes  p_i \]
Finally, as $p_i$ is zero dimensional, $p_i\boxtimes \mathcal A\cong \mathcal A$ and therefore
\[\Omega_\mathcal{S'}\cong \bigoplus\limits_{i=1}^N\mathcal A_i = \Omega_\mathcal{S} \]
That is, $\mathcal S$ is just the composite of some possibly non-classical state space $\mathcal{A}$ with a classical state space $\Delta_N$ where the state of $\Delta_N$ labels which state $s_i$ a state decoheres to. Therefore, the output state of the decoherence process depends only on the state of the classical component. Hence, can be written in the required form,
\[D_\mathcal{S} = (x\circ u)_\mathcal{A}\otimes\mathds{1}_{\Delta_N}.\]
where $u$ is the discarding effect and $x$ is some fixed state of $\mathcal{A}$.
\endproof

\end{document}